\documentclass[runningheads,a4paper]{llncs}
\usepackage{amssymb}
\usepackage{amsmath}
\usepackage{graphicx}
\usepackage{booktabs}
\usepackage{polynom}
\usepackage{esvect}
\usepackage{tikz}
\usepackage{blkarray}
\usepackage{mathdots}
\usepackage{enumitem}
\usepackage{mathrsfs}
\DeclareMathAlphabet{\mathpzc}{OT1}{pzc}{m}{it}
\usetikzlibrary{matrix,calc}

\makeatletter
\def\Ddots{\mathinner{\mkern1mu\raise\p@
\vbox{\kern7\p@\hbox{.}}\mkern2mu
\raise4\p@\hbox{.}\mkern2mu\raise7\p@\hbox{.}\mkern1mu}}
\makeatother
\makeatletter
\renewcommand{\maketag@@@}[1]{\hbox{\m@th\normalsize\normalfont#1}}%
\makeatother
\usepackage{multirow}
\usepackage[utf8]{inputenc}
\usepackage[font={small}]{caption}
\usepackage{float}
\restylefloat{table}
\usepackage{bm}
\usepackage{makecell}
\usepackage{boldline}
\setcellgapes{3pt}
\usepackage{multicol}
\usepackage{newfloat}
\usepackage{caption}
\usepackage{algorithm}
\usepackage{scalerel}
\usepackage{isomath}
\usepackage{tikz}
\newcommand{\framedbox}[2][0.95\textwidth]{
  \centering
  \tikzstyle{mybox} = [draw=black,line width=1.2pt,inner sep=8pt]
  \begin{tikzpicture}
   \node [mybox] (fig){%
    \begin{minipage}{#1}{#2}\end{minipage}
   };
  \end{tikzpicture}
}

\DeclareFloatingEnvironment[listname=Box,name=Fig]{story}
\usepackage{algpseudocode}
\algrenewcommand\alglinenumber[1]{\tiny #1:}
\algdef{SE}{Begin}{End}{\textbf{begin}}{\textbf{end}}

\usepackage{etoolbox}

\newcommand{\algstrut}[1][\algruledefaultfactor]{\vrule width 0pt
depth .25\baselineskip height #1\baselineskip\relax}
\newcommand*{\algrule}[1][\algorithmicindent]{\hspace*{.5em}\vrule\algstrut
\hspace*{\dimexpr#1-.5em}}

\makeatletter
\newcount\ALG@printindent@tempcnta
\def\ALG@printindent{%
    \ifnum \theALG@nested>0
    \ifx\ALG@text\ALG@x@notext
    \else
    \unskip
    \ALG@printindent@tempcnta=1
    \loop
    \algrule[\csname ALG@ind@\the\ALG@printindent@tempcnta\endcsname]%
    \advance \ALG@printindent@tempcnta 1
    \ifnum \ALG@printindent@tempcnta<\numexpr\theALG@nested+1\relax
    \repeat
    \fi
    \fi
}%

\patchcmd{\ALG@doentity}{\noindent\hskip\ALG@tlm}{\ALG@printindent}{}{\errmessage{failed to patch}}

\AtBeginEnvironment{algorithmic}{\lineskip0pt}

\newcommand{\norm}[1]{\left\lVert#1\right\rVert}
\newcommand{\abs}[1]{\left\lvert#1\right\rvert}

\makeatletter
\newcommand{\oset}[3][0ex]{%
  \mathrel{\mathop{#3}\limits^{
    \vbox to#1{\kern-2\ex@
    \hbox{$\scriptstyle#2$}\vss}}}}
\makeatother

\begin{document}
\sloppy
\title{ A Somewhat Homomorphic Encryption Scheme based on Multivariate Polynomial Evaluation  }
\author{Uddipana Dowerah and Srinivasan Krishnaswamy \vspace{-2ex}
\institute{ Indian Institute of Technology, Guwahati \\
\email{\{d.uddipana, srinikris\}@iitg.ac.in}}}

\maketitle

\begin{abstract}
We propose a symmetric key homomorphic encryption scheme based on the evaluation of multivariate polynomials over a finite field. The proposed scheme is somewhat homomorphic with respect to addition and multiplication. Further, we define a generalization of the Learning with Errors problem called the Hidden Subspace Membership problem and show that the semantic security of the proposed scheme can be reduced to the hardness of this problem. 

\keywords
{Homomorphic Encryption $\cdot$ Multivariate Polynomials $\cdot$  Learning with Errors $\cdot$ Hidden Subspace Membership  }

\end{abstract}
\section{Introduction}
An encryption scheme is said to be homomorphic for a set of functions if these functions can be evaluated on the plaintexts from the corresponding ciphertexts without decryption. A scheme is called fully homomorphic if functions with an arbitrary number of additions and multiplications can be homomorphically evaluated. If the number of additions and multiplications is limited, then the scheme is said to be somewhat homomorphic. Homomorphic encryption has wide applications in multiparty computation \cite{cramer2001multiparty}, secure electronic voting \cite{benaloh1987verifiable}, private information retrieval \cite{lipmaa2003diophantine} etc. 

 The first feasible construction of a fully homomorphic encryption scheme was proposed in \cite{gentry2009fully}. This scheme is based on {ideal lattices} and uses the concept of bootstrapping to convert a somewhat homomorphic scheme to a fully homomorphic one. This concept was subsequently used in many other constructions \cite{brakerski2011fully,gentry2011fully,smart2010fully,van2010fully}. As an alternative to bootstrapping, a number of schemes were proposed \cite{brakerski2012fully,brakerski2014leveled,brakerski2014efficient,gentry2010simple,gentry2013homomorphic} based on the Learning with Errors (LWE) problem \cite{regev2009lattices,regev2010learning}.
 
 A multivariate polynomial based encryption scheme depends on the difficulty of solving a system of non-linear equations over a finite field. Polly Cracker schemes \cite{fellows1993combinatorial,barkee1994you} were attempts at constructing homomorphic encryption schemes using multivariate polynomials. The security of these schemes depends on the Ideal Membership Problem and relies on the hardness of computing a Gr\"{o}bner basis for an ideal. These schemes, however, are vulnerable to Linear Algebra-based attacks and attacks using efficient Gr\"{o}bner basis construction algorithms \cite{barkee1994you,levy2009survey}. Attempts have been made to overcome these vulnerabilities with the addition of noise as proposed in \cite{albrecht2011polly,herold2012polly}. 
 
 In this paper, we propose a multivariate polynomial based homomorphic encryption scheme. The encryption process involves noisily evaluating a multivariate polynomial at a set of secret points and adding the scaled plaintext bit to each of these evaluations. The polynomial is randomly chosen from a secret ideal. In order to prove the security of the scheme, we have defined a generalization of the Learning with Errors (LWE) problem \cite{regev2009lattices,regev2010learning} called the Hidden Subspace Membership (HSM) problem and have shown that the semantic security of the proposed scheme can be reduced to the hardness of this problem. 
 
  The remainder of this paper is organized as follows: Section \ref{sec2} contains the preliminaries for the proposed work. Section \ref{Construction} contains the construction of the proposed scheme.  In section \ref{HP}, we discuss the homomorphic properties of the scheme and section \ref{security} deals with the analysis of its security.

\section{Preliminaries} \label{sec2}

The following notations are used in the paper. $\lambda$ denotes the security parameter. $\mathbb{Z}$ and $\mathbb{N}$ denote the set of integers and natural numbers respectively. Given a set ${S}$, $x\oset{\scriptscriptstyle{\$}}{\leftarrow}\mathcal{S}$ means that $x$ is sampled  uniformly at random from ${S}$. For a real number $x$, $\lfloor x \rfloor$, $\lceil x \rceil $ and $\lfloor x \rceil$ denote the rounding of $x$ down, up or to the nearest integer of $x$. $\mathbb{Z}_q$ denotes the set of integers in the interval $(-\frac{q}{2},\frac{q}{2}]$. $\mathbb{F}_{q}$ denotes a finite field of cardinality $q$, where $q$ is a prime integer. It represents the elements in the interval \(\left[-\left\lfloor \frac{q}{2} \right\rfloor,\left\lfloor \frac{q}{2} \right\rfloor \right] \). $\mathbb{F}_q[x_1,\hdots,x_\ell]_{\leq r}$ denotes the set of polynomials in $x_1,\hdots,x_\ell$ with coefficients in $\mathbb{F}_q$ of degree $\leq r$. Given a set of points $\{\bm{z}_1,\hdots,\bm{z}_n\} \in \mathbb{F}_q^\ell$ and a polynomial $f \in \mathbb{F}_q[x_1,\hdots,x_\ell]_{\leq r}$, $f(\bm{z}_1,\hdots,\bm{z}_n) \in \mathbb{F}_q^n$ denotes the evaluation of $f$ at these points. We use uppercase bold letters $\bm{A}, \bm{B}, \hdots$ to denote matrices and lowercase bold letters $\bm{a}, \bm{b},\hdots$ to denote vectors. The $\ell_i$ norm of a vector $\bm{v}$ is denoted by \(\norm{\bm{v}}_i\) and the inner product of two vectors $\bm{v}_1, \bm{v}_2$ is denoted using $\langle \bm{v}_1,\bm{v}_2\rangle:=\bm{v}_1^T \bm{v}_2\). A function $f(x): \mathbb{N} \rightarrow \mathbb{R}$ is called negligible if, for every $c \in \mathbb{N}$, there exists an integer $n_c$ such that $\vert f(x) \vert < \frac{1}{x^c}$ for all $x > n^c$.

\subsection{Hidden Subspace Membership}
The Hidden Subspace membership (HSM) problem involves distinguishing a random vector from a noisy vector of a given subspace $\mathcal{S}$. We use the game playing framework adopted from \cite{bellare1997concrete} to formally define the HSM problem. 
\begin{definition}
{\bf(Hidden Subspace Membership)}. Let $\mathcal{S}$ be an $l$ dimensional subspace of the vector space $\mathcal{V}:=\mathbb{F}_q^n$, for some $n,q \in \mathbb{N}$ and $l \geq 1$, such that $q$ is prime and let $\mathcal{N}$ be a noise distribution on $\mathbb{F}_q^n$. Then, the HSM problem denoted as \textup{HSM}$_{n,q,\mathcal{N}}$ can be defined in terms of the game shown in Figure \ref{box:box2}. A PPT adversary $\mathcal{A}$ wins the game if it can guess the value of $\beta$ with a non-negligible advantage, given by
{ \begin{align*} 
Adv_{n,q,\mathcal{N}}^{\text{HSM},\mathcal{A}}(\lambda):=\Bigg| Pr[\text{HSM}_{n,q,\mathcal{N}}^{\mathcal{A}}(\lambda)\Rightarrow 1]-\frac{1}{2} \Bigg|
\end{align*}}
\end{definition}
 \begin{story}[!ht] \scriptsize
  \framedbox{\begin{multicols}{4} 
  \setlength{\columnsep}{.5cm}
\underline{\textbf{ Initialize}} 
\begin{algorithmic} \scriptsize
\Begin
\State $\!\mathcal{S} \gets {\mathcal{V}}(\lambda,l)$
\State $\beta \oset{\scriptscriptstyle{\$}}{\leftarrow} \{0,1\} $
\End  \\~\\~\\~\\~\\
\end{algorithmic}

\underline{\textbf{ Sample}( )} 
\begin{algorithmic} \scriptsize
\Begin
\State $\bm{v} \oset{\scriptscriptstyle{\$}}{\leftarrow} \mathcal{S}$, $\bm{e} \oset{\scriptscriptstyle{\$}}{\leftarrow} \mathcal{N} $
\State $\bm{v} \gets\bm{v}+ \bm{e}$
\State \Return $\bm{v}$
\End \\~\\~\\~\\
\end{algorithmic} 

\underline{\textbf{ Challenge}( )} 
\begin{algorithmic} \scriptsize
\Begin
\State $\bm{v} \oset{\scriptscriptstyle{\$}}{\leftarrow} \mathbb{F}_{q}^{n}$
\If {$\beta =1 $}
\State    $ \bm{v} \oset{\scriptscriptstyle{\$}}{\leftarrow} \mathcal{S}$
\State $\bm{e} \oset{\scriptscriptstyle{\$}}{\leftarrow} \mathcal{N}$
\State $\bm{v} \leftarrow \bm{v}+\bm{e}$
\EndIf
\State \Return $\bm{v}$
\End 
\end{algorithmic}

\underline{\textbf{ Finalize ($\beta^{\prime}$)}} 
\begin{algorithmic} \scriptsize
\Begin 
\State \Return $(\beta \!=\! \beta^{\prime})$
\End \\~\\~\\~\\~\\
\end{algorithmic}
  \end{multicols}}
\caption[]{HSM$_{n,q,\mathcal{N}}$ Game}
  \label{box:box2}
\end{story}

This problem can be seen as a generalization of the Decisional Learning with Errors (DLWE) problem which is defined as follows: 
\begin{definition}
{\bf(Decisional Learning With Errors)}.  Given a noise distribution $\mathcal{X}$ on $\mathbb{Z}_q$ and an integer $q=q(\lambda)$, the DLWE problem can be defined in terms of the game shown in Figure \ref{box:box3}. A PPT adversary $\mathcal{A}$ wins the game if it can guess the value of $\beta$ with a non-negligible advantage given by 
\end{definition}
{\footnotesize \begin{align*}
Adv_{n,q,\mathcal{X}}^{\text{LWE},\mathcal{A}}(\lambda):=\Bigg|Pr[\text{LWE}_{n,q,\mathcal{X}}^{\mathcal{A}}(\lambda) \Rightarrow 1]-\frac{1}{2} \Bigg|
\end{align*}}%

\vspace{-3ex}
\begin{story}[H] \scriptsize
  \framedbox{\begin{multicols}{4} 
  \setlength{\columnsep}{.55cm}
  
\underline{\textbf{ Initialize}} 
\begin{algorithmic} \scriptsize
\Begin
\State $n\gets $ $n(\lambda)$
\State $\bm{s} \oset{\tiny{\$}}{\leftarrow} \mathbb{Z}_q^{n}$
\State $\beta \oset{\scriptscriptstyle{\$}}{\leftarrow} \{0,1\} $
\End \\~\\~\\~\\~\\
\end{algorithmic}

\underline{\textbf{ Sample}( )} 
\begin{algorithmic} \scriptsize
\Begin
\State $\bm{a} \oset{\scriptscriptstyle{\$}}{\leftarrow} \mathbb{Z}_q^{n}$
\State $ {e} \oset{\scriptscriptstyle{\$}}{\leftarrow} \mathcal{X} $
\State $b \!= \left[ \bm{a}\!^{T} \!\!\bm{s} \!+\!{e}\right]_q$
\State \Return $(\bm{a},b)$
\End \\~\\~\\
\end{algorithmic}

\underline{\textbf{ Challenge}( )} 
\begin{algorithmic} \scriptsize
\Begin
\State $(\bm{a},b) \oset{\scriptscriptstyle{\$}}{\leftarrow} \mathbb{Z}_{q}^{n+1}$
\If {$\beta =1 $}
\State    $ \bm{a} \oset{\scriptscriptstyle{\$}}{\leftarrow} \mathbb{Z}_q^{n}$
\State ${e} \oset{\scriptscriptstyle{\$}}{\leftarrow} \mathcal{X}$
\State $b \!= \!\left[\! \bm{a}\!^{T} \!\!\bm{s} \!+\!{e}\!\right]_q$
\EndIf
\State \Return $(\bm{a},b)$
\End 
\end{algorithmic}

\underline{\textbf{ Finalize ($\beta^{\prime}$)}} 
\begin{algorithmic} \scriptsize
\Begin
\State \!\!\Return $(\beta \!=\! \beta^{\prime})$
\End \\~\\~\\~\\~\\~\\
\end{algorithmic}
  \end{multicols}}
\caption[]{DLWE$_{n,q,\mathcal{X}}$ Game}
  \label{box:box3}
\end{story}

An LWE sample $(\bm{a}_i,b_i)\in \mathbb{Z}_q^n \times \mathbb{Z}_q$ that satisfies $\bm{a}_i^{T} \bm{s} +e_i =b_i~(\text{mod}~q)$ can be written as: 
{\small \begin{align} \label{lwe}
\begin{bmatrix} \bm{a}_i & -{b}_i \end{bmatrix}\begin{bmatrix}
\bm{s} \\ 1
\end{bmatrix} \approx_{\scriptscriptstyle \mathcal{X}} {0}~(\text{mod}~q)
\end{align}}%
 which is a noisy equation with the noise being sampled from the distribution $\mathcal{X}$ on $\mathbb{Z}_q$. Hence, the DLWE problem can be seen as an instance of the HSM problem where the subspace is $(\bm{s}~1)^{\perp} \subseteq \mathbb{Z}_q^{n}$ and the noise is added only to the last element of the vector. The above discussion is formalized in terms of the following lemma. 

\begin{lemma}
Let $\mathcal{S}$ be an $n$-dimensional subspace of the vector space $\mathbb{Z}_q^{n+1}$. Let $\mathcal{X}$ be a zero mean distribution on $\mathbb{Z}_q$ and let $\mathcal{N}:=(\bm{0} ~\mathcal{X})$, where $\bm{0} \in \mathbb{Z}_q^n$. Then, any PPT adversary $\mathcal{A}$ against the HSM problem \textup{HSM}$_{n+1,q,\mathcal{N}}$ can be transformed into a PPT adversary $\mathcal{B}$ against the LWE problem \textup{DLWE}$_{n,q,\mathcal{X}}$ such that
\end{lemma}
\vspace{-3ex}
{\footnotesize \begin{align}
Adv_{n+1,q,\mathcal{N}}^{\text{HSM},\mathcal{A}}(\lambda)=Adv_{n,q,\mathcal{X}}^{\text{DLWE},\mathcal{B}}(\lambda) \notag
\end{align}}%

\begin{proof}
  We construct an adversary $\mathcal{B}$ against \textup{DLWE}$_{n,q,\mathcal{X}}$ from an adversary $\mathcal{A}$ against \textup{HSM}$_{n+1,q,\mathcal{N}}$. Observe that, an HSM sample $\bm{v}_i\leftarrow \bm{v}_i+\bm{e}_i$, where $\bm{v}_i \oset{\scriptscriptstyle{\$}}{\leftarrow} \mathcal{S}$ and $\bm{e}_i \oset{\scriptscriptstyle{\$}}{\leftarrow} \mathcal{N}$, satisfies the equation $\bm{v}_i^T \bm{s}_1 \approx_{\scriptscriptstyle {\mathcal{X}}} 0$, where $\bm{s}_1:=(\bm{s},1) \in \mathbb{Z}_q^{n+1}$. 

When $\mathcal{A}$ queries the \textbf{Sample} oracle of HSM, $\mathcal{B}$ queries the \textbf{Sample} oracle of DLWE to obtain $(\bm{a},b)$ and returns the vector $(\bm{a},-b) \in \mathbb{Z}_q^{n+1}$ to $\mathcal{A}$. Similarly, when $\mathcal{A}$ queries the \textbf{Challenge} oracle of HSM, $\mathcal{B}$ queries the \textbf{Challenge} oracle of DLWE and answers similarly. If $\mathcal{A}$ outputs the correct value of $\beta$ for \textup{HSM}$_{n+1,q,\mathcal{N}}$, then $\mathcal{B}$ also happens to output the correct value of $\beta$ for \textup{DLWE}$_{n,q,\mathcal{X}}$. 

Thus solving the LWE problem over $\mathbb{Z}_q^n$ is equivalent to solving the HSM problem over $\mathbb{Z}_q^{n+1}$ for a subspace of dimension $n$, where only the last entry is affected by noise. 
\end{proof}
\begin{definition}
{\bf(Discrete Gaussian distribution)}. For $\alpha >0$ and for $q=q(\lambda)$, the density function of a Gaussian distribution $ \mathrm{\Psi}_{\alpha} $ over $\mathbb{R}$ is given by, 
{\footnotesize 
$D_{\alpha}(x)=\frac{1}{\sqrt{2\pi}\alpha}exp\left(-\frac{1}{2}\left(\frac{x}{\alpha}\right)^2\right)$}. Then, the discrete Gaussian distribution $\overline{\mathrm{\Psi}}_{\alpha}$ on $\mathbb{Z}_q$ is the distribution with mean zero and standard deviation $\alpha q$, obtained by sampling an element $x \leftarrow D_{\alpha}$ and outputting $\lfloor x \cdot q \rceil ~\text{mod}~q$. 
\end{definition}
When the noise is sampled according to the distribution $\overline{\mathrm{\Psi}}_{\alpha}$ with $\alpha q \geq 2 \sqrt{n}$, there exists a quantum reduction from the decisional variant of the shortest vector problem (G{\footnotesize{AP}}SVP) to the LWE problem \cite{regev2009lattices}. A classical reduction has also been proposed in \cite{peikert2009public}.

We assume that the HSM problem is hard for an $l$-dimensional subspace of an $n$-dimensional space when $n-l$ entries are corrupted by noise sampled from a discrete Gaussian distribution. (The justification for this assumption and a detained analysis of the hardness of the HSM problem will be given in the complete paper.)

\section{The Proposed Scheme} \label{Construction}

Consider an ideal $I $ of the polynomial ring $ \mathbb{F}_q[x_1,\hdots,x_\ell]$, where $q$ is polynomial in the security parameter $\lambda$. Let $I_{\leq r}$ denotes the set of polynomials in $I$ with degree $\leq r$. $ \mathbb{F}_q[x_1,\hdots,x_\ell]_{\leq r} $ is a vector space over $\mathbb{F}_q$ of dimension $N:=\binom{\ell+r}{r}$ and $I_{\leq r}$ is a subspace of this vector space. Let $ dim(I_{\leq r}) $ be the dimension of this subspace. A polynomial $f \in \mathbb{F}_q[x_1,\hdots,x_\ell]_{\leq r}$ evaluated at all points of $ \mathbb{F}_q^{\ell} $ generates a vector space in $ \mathbb{F}_q^{\scaleto{q^\ell}{6.8pt}} $. The set of such vectors obtained by evaluating all polynomials in $ \mathbb{F}_q[x_1,\hdots,x_\ell]_{\leq r} $ constitutes an $N$-dimensional subspace of $ \mathbb{F}_q^{\scaleto{q^\ell}{6.8pt}} $. Evaluating polynomials in $I_{\leq r}$ gives us a $ dim(I_{\leq r}) $-dimensional subspace $\mathcal{V}_{\scriptstyle I_{\leq r}}$. The ciphertext size $n$ is chosen such that $n $ is polynomial in $\lambda$ but is significantly less than $q$ and $dim(I_{\leq r}) < n \leq N$. We then choose $n$ distinct points $ \{\bm{z}_1,\hdots,\bm{z}_n\} \in \mathbb{F}_q^{\ell} $ which satisfy the following conditions:
\begin{enumerate}
\item Every vector in  $ \mathbb{F}_q^n $ can be got by evaluating a polynomial in $ \mathbb{F}_q[x_1,\hdots,x_\ell]_{\leq r} $ at $ (\bm{z}_1,\hdots,\bm{z}_n )$.
\item Every vector in $\mathbb{F}_q^{\scaleto{dim(I_{\leq r})}{6.8pt}}$ can be got by evaluating a polynomial in $I_{\leq r}$ at $ (\bm{z}_1,\hdots,\bm{z}_{\scaleto{dim(I_{\leq r})}{6.5pt}} )$.
\end{enumerate} 
The conditions imposed on $ (\bm{z}_1,\hdots,\bm{z}_n )$ ensure that for every given vector $\bm{s}_2 \in \mathbb{F}_q^{n-\scaleto{dim(I_{\leq r})}{6.5pt}}$, there exists a vector $\bm{s} \in (\mathcal{V}_{\scriptstyle I_{\leq r}})^{\perp}$ such that $\bm{s}=(\bm{s}_1,\bm{s}_2)$.

 Let $\bm{v}_i \in \mathbb{F}_q^{\scaleto{ N}{4pt}}$ be the vector obtained by evaluating the monomials of degree $\leq r$ at the point $\bm{z}_i \in \mathbb{F}_q^\ell$, for $1 \leq i \leq n$. Let $\bm{G} \in \mathbb{F}_q^{n \times \scaleto{ N}{4pt}}$ be the matrix obtained by assigning $\bm{G}(i,:)=\bm{v}_i$ for $1 \leq i \leq n$. Let $f(x_1,\hdots,x_\ell)$ be a polynomial, sampled uniformly at random from $I_{\leq r}$. If $\bm{f} \in \mathbb{F}_q^{\scaleto{ N}{4pt}}$ denotes the coefficient vector of $f$, then \(f(\bm{z}_1,\hdots,\bm{z}_n)=\bm{G} \cdot \bm{f} \in \mathbb{F}_q^n \).

 Let us now consider an example to illustrate the construction of $\bm{G}$. For $\ell=2,r=2$ and $N=6$, the monomials of $\mathbb{F}_q[x_1,x_2]_{\leq 2}$ according to the lexicographic order are $\{1,x_1,x_2,x_1^2,x_1x_2,x_2^2\}$. If $n=3$, then let $\bm{z}_1=(z_{11},z_{21}), \bm{z}_2=(z_{12},z_{22}), \bm{z}_3=(z_{13},z_{23})$ be 3 distinct points in $\mathbb{F}_q^{2}$. Then, $G \in \mathbb{F}_q^{3 \times 6}$ represents the following matrix:
{ \begin{align}
\bm{G} = \begin{bmatrix}
   1 & z_{11} & z_{21} & z_{11}^2 & z_{11}z_{21} & z_{21}^2\\
   1 & z_{12} & z_{22} & z_{12}^2 & z_{12}z_{22} & z_{22}^2\\
   1 & z_{13} & z_{23} & z_{13}^2 & z_{13}z_{23} & z_{23}^2\\
\end{bmatrix}
\end{align}}%

The proposed encryption scheme consists of the following algorithms. The plaintext space is $\{0,1\}$ and the ciphertexts are vectors of order $n$ over $\mathbb{F}_q$. Let $\mathcal{X}$ be a noise distribution that samples its entries from the discrete Gaussian distribution $\overline{\mathrm{\Psi}}_{\alpha}$ on $\mathbb{Z}_q$. 
  \begin{itemize}[leftmargin=*]
\setlength\itemsep{1.5ex}
\item[$  \bullet$] \textbf{KeyGen}($1^{\lambda}$): Consider an ideal $I \subseteq \mathbb{F}_q[x_1,\hdots,x_\ell]$ and a set of $n$ distinct points $ \{\bm{z}_1,\hdots,\bm{z}_n \}  $ in $\mathbb{F}_q^\ell $ and construct the matrix $\bm{G} \in \mathbb{F}_q^{n \times \scaleto{ N}{4pt}}$.  
Choose a vector $\bm{s}=(\bm{s}_1,\bm{s}_2) \in (\mathcal{V}_{\scriptstyle I_{\leq r}})^{\perp}$ and a constant $p \in \mathbb{N}$ such that, if { \(\sigma_{\scriptstyle \bm{s}}:=\sum_{i=1}^{n} s_{i} \)} denotes the sum of the elements of $\bm{s}$ and $\bm{e}$ denotes a vector of order $n$ whose entries are sampled from the distribution $\mathcal{X}$, then $\eta(\alpha q) < \sigma_{\!\scriptstyle \bm{s}}\cdot p \leq \lfloor q/2 \rfloor $ and the distribution of $\langle \bm{s},\bm{e} \rangle$ has a standard deviation $< k(\alpha q)$, where $\eta$ and $k$ are constants that depend on the probability of decryption error and the number of homomorphic operations. 
The secret key is a basis for the ideal $I$, the set of points $ \{\bm{z}_1,\hdots,\bm{z}_n \} $ and the integer $p$. If the choice of $I$ is restricted to a radical ideal whose variety has a finite number of points, then the ideal can also be expressed in terms of its variety.

\item[$  \bullet$] \textbf{Encrypt}$(sk,{m})$:   To encrypt a message $m \in \{0,1\}$, sample a polynomial $f$ of degree $\leq r$, uniformly at random from the ideal $I$ and a vector $\bm{e}=(\bm{0},\overline{\bm{e}})$, where $\overline{\bm{e}}$ is a vector of order $n-\scaleto{dim(I_{\leq r})}{10pt}$ and the entries of $\overline{\bm{e}}$ are chosen according to the distribution $\mathcal{X}$. If $\bm{f}$ denotes the coefficient vector of $f$ and $\bm{1}$ denotes the all 1's vector of order $n$ and $M:=m \cdot p$, then the ciphertext can be computed as
{ \begin{align}
\textstyle \bm{c}=\textstyle  M  \cdot \bm{1}+\bm{G} \cdot \bm{f}+  \bm{e} ~~(\text{mod}~q)
\end{align}}%
\item[$  \bullet$] \textbf{Decrypt}$(sk,\bm{c})$: 
Given the ciphertext $\bm{c}$, the plaintext $m$ can be recovered as 
{\small \begin{align}
 {m}=\left\lfloor \frac{1}{\sigma_{\scriptstyle \bm{s}}\cdot p}\left(\langle \bm{s},\bm{c}\rangle \,\text{mod}\,q\right)\right\rceil\,\text{mod}~2 
\end{align} }%
   \end{itemize}
   
\begin{remark}
This scheme can also be seen as a coding theory based scheme. If the polynomial chosen in the above encryption scheme is $f$, then the ciphertext can be seen as a noisy punctured  Reed-Muller encoding of the polynomial $f+m$.
\end{remark} 

\subsubsection{Correctness of Decryption.} \label{COD}
Since $\bm{s} \in (\mathcal{V}_{\leq r})^{\perp}$, \( \langle \bm{s},\bm{G} \bm{f} \rangle=0~(\text{mod}~q)\). Therefore,
\begin{align}
 \langle \bm{s},\bm{c}\rangle &= m\cdot (\sigma_{\scriptstyle \bm{s}} p)+\langle \bm{s}_2,\overline{\bm{e}}\rangle ~~(\text{mod}~q)
\end{align}
If {\small$ \abs{\langle \bm{s}_2,\overline{\bm{e}}\rangle}  < \left\lfloor\frac{(\sigma_{ \bm{s}}  \cdot p)}{2}\right\rfloor $}, then {\small \( \left\lfloor \frac{1}{(\sigma_{\scriptstyle \bm{s}}\cdot p)}\left(\langle \bm{s},\bm{c}\rangle \,\text{mod}\,q\right)\right\rceil = m~(\text{mod}~2)\)}. If for some $ \epsilon>0 $, it holds that %
{ \begin{align}
\textstyle \underset{{\overline{\bm{e}} {\leftarrow}\mathcal{X}^{n-\scaleto{dim(I_{\leq r})}{6pt}}}}{Pr}\left[ \abs{\langle \bm{s}_2,\overline{\bm{e}}\rangle}  > \left\lfloor\frac{(\sigma_{ \bm{s}}  \cdot p)}{2}\right\rfloor\right] \leq \textstyle \epsilon
\end{align}}
then, the probability of decryption error is at most $ \epsilon $. In other words, the decryption function will output the correct message with probability $1-\epsilon$. Since the entries of $\overline{\bm{e}}$ are sampled from the distribution $\mathcal{X}$, $ \langle \bm{s}_2,\overline{\bm{e}} \rangle $ is distributed with a standard deviation $\norm{\bm{s}_2}_2\alpha q$. If $\mathcal{X}_{\scaleto{\norm{\bm{s}_2}_2\alpha q}{6pt}}$ denotes this distribution, then using Chebyshev's inequality, we can write 
{ \begin{align}
\textstyle \underset{{x {\leftarrow}\mathcal{X}_{\scaleto{\norm{\bm{s}_2}_2\alpha q}{5.2pt}}}}{Pr}\left[\abs{x} > \left\lfloor\frac{(\sigma_{\scriptstyle \bm{s}}\cdot p)}{2} \right\rfloor\right] \leq \textstyle \underset{{x {\leftarrow}\mathcal{X}_{\scaleto{\norm{\bm{s}_2}_2\alpha q}{5.2pt}}}}{Pr}\left[\abs{x} > k\alpha q \right] \leq \frac{1}{k^2}
\end{align}}%
where  \(k < \frac{\sigma_{\bm{s}}  \cdot p}{2\norm{\bm{s}_2}_2\alpha q} $ and $\epsilon=\frac{1}{k^2} $. Therefore, if $p$ is chosen such that $ \sigma_{\bm{s}}\cdot p > (2\norm{\bm{s}_2}_2 \alpha q)/\sqrt{\epsilon}$, then the probability of decryption error is $\leq \epsilon$. 
 
 For the case when only additive homomorphism is required, we can choose $\bm{s}=(\bm{s}_1,\bm{s}_2) \in (\mathcal{V}_{\scriptstyle I_{\leq r}})^{\perp}$ such that $\bm{s}_2$ is of the form $\bm{s}_2=(0,\hdots,0,1)$. Thus $\langle \bm{s},\bm{e} \rangle$ will just be equal to the last entry of $\bm{e}$. However as we will see in the subsequent section that this may not be possible when multiplicative homomorphism is required.

\section{Homomorphic Properties}\label{HP}
 
  If \(\phi:\{0,1\}^t \rightarrow \{0,1\}$ denotes a function to be performed on the plaintexts $m_1,\hdots,m_t$, then homomorphically evaluating $\phi$ involves calculating a new ciphertext \(\bm{c}_{\scriptscriptstyle eval} \) such that \(Dec(\bm{c}_{\scriptscriptstyle eval},sk)=\phi(m_1,\hdots,m_t)\). Any function $\phi$ involving binary variables can be evaluated using a set of addition and multiplication gates. The remainder of this section analyses homomorphic addition and multiplication of ciphertexts in the proposed scheme.

\subsection{Addition.} Given two ciphertexts $\bm{c}_1$ and $\bm{c}_2$ of the respective plaintexts $m_1$ and $m_2$, where \(\bm{c}_1=M_1 \cdot \bm{1}+\bm{G} \cdot \bm{f}_1+\bm{e}_1\) and \(\bm{c}_2=M_2 \cdot \bm{1}+\bm{G} \cdot \bm{f}_2+\bm{e}_2\), homomorphic addition of $m_1$ and $m_2$ can be performed by computing 
{ \begin{align}
 \bm{c}_{\scriptscriptstyle add} &=\bm{c}_1+  \bm{c}_2~~(\text{mod}~q) \notag \\
  & =(M_1 + M_2)\cdot \bm{1}+\bm{G}  (\bm{f}_1+\bm{f}_2)+(\bm{e}_1+\bm{e}_2)~~(\text{mod}~q)
\end{align}}%
Since $ \langle \bm{s}, \bm{G}  (\bm{f}_1+\bm{f}_2) \rangle =0~(\text{mod}~q)$ and given that, $M_i=m_i \cdot p$ and $\bm{e}_i=(\bm{0},\overline{\bm{e}}_i)$ for $i \in \{1,2\}$, where $\overline{\bm{e}}_i$ is a vector of order $n-\scaleto{dim(I_{\leq r})}{10pt}$ and its entries are sampled from the distribution $\mathcal{X}$, we have
\begin{align}
\langle \bm{s},\bm{c}_{\scriptscriptstyle add}\rangle =(m_1+ m_2) \cdot \sigma_{\!\scriptstyle \bm{s}}\, p+ \langle \bm{s}_2,\overline{\bm{e}}_1+\overline{\bm{e}}_2 \rangle~(\text{mod}~q)
\end{align}
 If \({e}_{\scriptscriptstyle add}:= \langle \bm{s}_2,\overline{\bm{e}}_1+\overline{\bm{e}}_2 \rangle~(\text{mod}~q)\), then $  \bm{c}_{\scriptscriptstyle add} $ decrypts correctly to $ (m_1+m_2)~\text{mod}\,2$ as long as \((m_1+ m_2) \cdot \sigma_{\!\scriptstyle \bm{s}}\, p  < \lfloor q/2 \rfloor\) and {\small \(\abs{{e}_{\scriptscriptstyle add}} < \left\lfloor {(\sigma_{\!\scriptstyle \bm{s}}\cdot p)}/{2} \right\rfloor \)}. The distribution of $  {e}_{\scriptscriptstyle add} $ is Gaussian with a standard deviation $ \sqrt{2}\norm{\bm{s}_2}_2\alpha q $. It can be easily verified that due to the increase in the standard deviation after addition, the probability of decryption error increases by a factor of 2.

\subsection{Multiplication.} Multiplication is based on the fact that, given two polynomials $f_1,f_2 \in \mathbb{F}_q[x_1,\hdots,x_\ell]_{\leq r}$ and an evaluation point $\bm{z} \in \mathbb{F}_q^\ell$, 
{\small \begin{align}
f_1(\bm{z}) \cdot f_2(\bm{z})=\left(f_1 \cdot f_2\right) (\bm{z})
\end{align}}%
where deg$(f_1f_2) \leq 2r$. Therefore, $f_1f_2 \in I$ but it is not an element of the subspace $I_{\leq r}$. In order to do homomorphic multiplication, $n$ must be greater than the dimension of $\mathcal{V}_{\scriptstyle I_{\leq 2r}}$ and the vector $\bm{s}=(\bm{s}_1,\bm{s}_2)$ must be an element of {\small$ (\mathcal{V}_{\scriptstyle I_{\leq 2r}} )^{\perp}$} such that $\bm{s}_2 \in \mathbb{F}_q^{n-\scaleto{dim(I_{\leq 2r})}{6.5pt}}$. The noise must be added to the last $n-\scaleto{dim(I_{\scriptstyle \leq r})}{10pt}$ entries.

Given two ciphertexts $\bm{c}_1$ and $\bm{c}_2$, the homomorphic multiplication of its respective plaintexts $m_1$ and $m_2$ can be obtained by taking the component-wise product of $\bm{c}_1$ and $\bm{c}_2$. In order to maintain the invariant structure \(\langle \bm{s},\bm{c} \rangle=m\cdot (\sigma_{\!\scriptstyle \bm{s}}\, p)+\langle \bm{s}_2,\overline{\bm{e}} \rangle\) for decryption, we multiply the product by \(1/p\). Specifically, if \(\bm{c}_1=M_1 \cdot \bm{1}+\bm{G} \cdot \bm{f}_1+\bm{e}_1\) and \(\bm{c}_2=M_2 \cdot \bm{1}+\bm{G} \cdot \bm{f}_2+\bm{e}_2\), then 

\vspace{-1ex}
{ \small \begin{align} 
\bm{c}_{\scriptscriptstyle mult} &=\frac{1}{p}\left(\bm{c}_1\odot \bm{c}_2\right) ~~(\text{mod}~q) \notag \\
&=(m_1 m_2)p\cdot \bm{1}+m_1\cdot \bm{G}\bm{f}_2+m_1 \cdot \bm{e}_2+m_2\cdot \bm{G}\bm{f}_1+m_2 \cdot \bm{e}_1 \notag \\ &\quad +\frac{1}{p}\left[\bm{G}\bm{f}_1 \odot \bm{e}_2+\bm{e}_1 \odot \bm{G}\bm{f}_2+(\bm{G}\bm{f}_1 \odot \bm{G}\bm{f}_2)+(\bm{e}_1 \odot \bm{e}_2)\right] ~~(\text{mod}~q)
\end{align}}%
where $\odot$ denotes the component-wise product of two vectors. Given that, $M_i=m_i \cdot p$ and $\bm{e}_i=(\bm{0},\overline{\bm{e}}_i)$ for $i \in \{1,2\}$, where $\overline{\bm{e}}_i$ is a vector of order $n-\scaleto{dim(I_{\leq r})}{10pt}$ and its entries are sampled from the distribution $\mathcal{X}$,  
{ \begin{align} \scaleto{
\langle \bm{s},\bm{c}_{\scriptscriptstyle mult}\rangle =(m_1 m_2) \cdot \sigma_{\!\scriptstyle \bm{s}}\, p+\left\langle \bm{s}_2, m_1 \cdot\overline{\bm{e}}_2+m_2 \cdot \overline{\bm{e}}_1+\frac{1}{p}\cdot(\overline{\bm{e}}_1\odot\overline{\bm{e}}_2) \right\rangle~(\text{mod}~q) }{22pt}
\end{align}}
If {\small \( {e}_{\scriptscriptstyle mult} := \left\langle \bm{s}_2, m_1 \cdot\overline{\bm{e}}_2+m_2 \cdot \overline{\bm{e}}_1+\frac{1}{p}\cdot(\overline{\bm{e}}_1\odot\overline{\bm{e}}_2) \right\rangle~(\text{mod}~q) \)}, then $ \bm{c}_{\scriptscriptstyle mult} $ decrypts correctly to $m_1  m_2$ as long as {\small \(\abs{{e}_{\scriptscriptstyle mult}}< \left\lfloor {(\sigma_{\!\scriptstyle \bm{s}}\cdot p)}/{2} \right\rfloor\)}, where $ {e}_{\scriptscriptstyle mult} $ is a random variable with mean zero and standard deviation $< \sqrt{2}\alpha q+\frac{1}{\sqrt{p}}(\alpha q)^2 $.

While homomorphic addition does not impose any constraints on the ciphertext size, it does lead to an increase in noise. This limits the number of additions possible. Homomorphic multiplication introduces more noise than addition and the size of the ciphertext increases significantly with the number of multiplications possible. The increase in noise by homomorphic operations can be countered by using techniques such as modulus switching \cite{brakerski2014efficient,brakerski2014leveled}. (The complete paper will contain a more detailed discussion about the same.) Countering the increase in ciphertext size for homomorphic multiplication remains a challenge.

\section{Security}  \label{security}
  We show that the proposed scheme is semantically secure based on the hardness of the Hidden Subspace Membership problem. In a Chosen Plaintext Attack (CPA) model, the adversary has a number of plaintext-ciphertext pairs at its disposal. A symmetric key encryption scheme is said to be semantically secure or indistinguishable under a chosen plaintext attack (IND-CPA) if, given sufficient samples (plaintext-ciphertext pairs), no adversary can distinguish between the encryptions of 0 from the encryptions of messages of its choice with probability more than $\frac{1}{2}$. 
\begin{definition}
{\bf{(IND-CPA Security)}}. The IND-CPA security of a symmetric encryption scheme can be defined in terms of the game shown in Figure \ref{box:box4}. 
A PPT adversary $\mathcal{A}$ selects two messages $(m_0,{m}_1)$ such that one of them is 0 and the Left-Right oracle outputs the encryption of one of the messages by choosing $\beta \oset{\scriptscriptstyle{\$}}{\leftarrow} \{0,1\}$. $\mathcal{A}$ wins the game if it can guess the value of $\beta$ with a non-negligible advantage defined by 
\end{definition}
\vspace{-1ex}
{\footnotesize \begin{align*}
   {Adv}^{\text {IND-CPA},\mathcal{A}}(\lambda) := \Bigg| Pr[{\mathcal{A}}^{\text{ IND-CPA}}_{\text{Enc}_{sk}(\cdot)}(\lambda)=1]-Pr[\mathcal{A}^{\text{ IND-CPA}}_{\text{Enc}_{sk}(0)}(\lambda)=1] \Bigg|
\end{align*}}%

\vspace{-2ex}
\begin{story}[!ht] \scriptsize
  \framedbox{\begin{multicols}{4} 
  \setlength{\columnsep}{.5cm}
\underline{\textbf{ Initialize}} 
\begin{algorithmic} \scriptsize
\Begin
\State $sk \!\!\leftarrow\!\!  { \textbf{KeyGen}()}$ 
\State $\beta \oset{\scriptscriptstyle{\$}}{\leftarrow} \{0,1\} $
\End 
\end{algorithmic}

  \underline{\textbf{ Encrypt$({0},sk)$} }
\begin{algorithmic} \scriptsize
\Begin
\State $\bm{c} \!\leftarrow\! {\textbf{Enc}({m},sk)}$
\State \Return ${\bm{c}}$
\End 
\end{algorithmic} 

\underline{\textbf{ Left-Right$({m}_0,{m}_1)$} }
\begin{algorithmic} \scriptsize
\Begin
\State $\bm{c} \!\leftarrow\!\!\! {\textbf{Enc}({m}_{\tiny\beta},\!sk\!)}$
\State \Return $\bm{c}$
\End 
\end{algorithmic}  

 \underline{\textbf{ Finalize ($\beta^{\prime}$)}} 
\begin{algorithmic} \scriptsize
\Begin 
\State \Return $(\beta \!=\!\beta^{\prime})$
\End 
\end{algorithmic}  
  \end{multicols}}
\caption[]{IND-CPA Game}
  \label{box:box4}
\end{story}

It is clear from the construction of the scheme that the encryptions of 0 are noisy elements of the subspace $\mathcal{V}_{\scriptstyle I_{\leq r}} \subseteq \mathbb{F}_q^n$. Therefore, distinguishing an encryption of 0 from the encryption of a random message is equivalent to solving the HSM problem as shown in the following theorem. 

\begin{theorem}
 A PPT adversary $\mathcal{A}$ that breaks the IND-CPA security of the proposed scheme with non-negligible advantage $\epsilon$ can be converted into a PPT adversary $\mathcal{B}$ that can solve an instance of the HSM problem with advantage at least $\frac{\epsilon}{2}$. 
\end{theorem}
\begin{proof}
$\mathcal{B}$ initializes $\mathcal{A}$ with the parameters $(q,n,\mathcal{X})$. When $\mathcal{A}$ asks for an encryption of $0$, $\mathcal{B}$ queries the procedure \textbf{Sample} of the HSM game to get $\bm{v} \leftarrow \bm{v}+ \bm{e}$, where $\bm{e}\oset{\scriptscriptstyle{\$}}{\leftarrow} \mathcal{X}^n$ and returns the vector $\bm{c}=\bm{v}$. Similarly, when $\mathcal{A}$ queries the \textbf{Left-Right} oracle of the IND-CPA game, $\mathcal{B}$ queries the procedure \textbf{Challenge} of the HSM game to get $\bm{v} $ and returns a vector $\bm{c}$ upon choosing $\beta \oset{\scriptscriptstyle{\$}}{\leftarrow}\{0,1\}$ and setting $\bm{c}=\bm{v}+(p\cdot m_\beta )\cdot \bm{1}$. 

If the sample obtained from the \textbf{Challenge} oracle of HSM is a noisy element of $\mathcal{S}$, then $\mathcal{A}$ runs in a similar environment to that of the IND-CPA game and hence, $\mathcal{B}$ outputs $\beta$ with probability $\frac{1}{2}+\epsilon$. On the other hand, if the sample returned is uniform in $\mathbb{F}_q^{ n}$, then $\mathcal{B}$ outputs $\beta$ with probability $\frac{1}{2}$. Therefore, the probability that $\mathcal{B}$ solves the HSM problem is $\geq \frac{1}{2}+\frac{\epsilon}{2}$. Hence, the advantage of $\mathcal{B}$ in solving the HSM problem is $\geq \frac{\epsilon}{2}$.
\end{proof}

\section{Conclusions}
A symmetric key homomorphic encryption scheme has been proposed based on the evaluation of multivariate polynomials. The security of the scheme is based on the hardness of the Hidden Subspace Membership problem. In its current form, the scheme is somewhat homomorphic and for a given ciphertext size, can perform a limited number of homomoprhic additions and multiplications. Future work includes converting the scheme into a fully homomorphic one and designing a public key variant of the same.




%
%

\bibliography{HE_ref}

\end{document}